%% file: censored.tex
\documentclass[aoas,preprint]{imsart}

\setattribute{journal}{name}{}

\RequirePackage[OT1]{fontenc}
\RequirePackage{amsthm,amsmath}
\RequirePackage{natbib}
\RequirePackage[colorlinks,citecolor=blue,urlcolor=blue]{hyperref}

\usepackage{algorithm, algpseudocode}
\usepackage{amsfonts, amsopn, amssymb}
  \numberwithin{equation}{section}
\usepackage{booktabs}
\usepackage{graphicx, epstopdf}

\startlocaldefs
\numberwithin{equation}{section}
\theoremstyle{plain}
\newtheorem{theorem}{Theorem}[section]

\newtheorem{corollary}[theorem]{Corollary}

\newtheorem{lemma}[theorem]{Lemma}

\input{macros}

\DeclareMathOperator{\unif}{unif}

\newcommand{\mubar}{\bar{\mu}}
\newcommand{\cK}{\mathcal{K}}
\newcommand{\cN}{\mathcal{N}}
\newcommand{\n}{\bar{n}}
\newcommand{\Sigmabar}{\bar{\Sigma}}

\newcommand{\V}{\mathcal{V}}
\newcommand{\X}{\skew{3}{\bar}{X}}
\newcommand{\y}{\bar{y}}
\endlocaldefs

\begin{document}

\begin{frontmatter}
\title{Valid post-correction inference for censored regression problems}
\runtitle{Valid post-correction inference}

\begin{aug}
\author{\fnms{Yuekai} \snm{Sun}\ead[label=e1]{yuekai@stanford.edu}}
\and
\author{\fnms{Jonathan E.} \snm{Taylor}\ead[label=e2]{jonathan.taylor@stanford.edu}}

\runauthor{Sun and Taylor}

\affiliation{Stanford University}

\address{Institute for Computational and Mathematical Engineering \\
Stanford University \\
475 Via Ortega, Stanford, CA 94305 \\
\printead{e1}}
\address{}

\address{Department of Statistics \\
Stanford University \\
390 Serra Mall, Stanford, CA 94305 \\
\printead{e2}}
\address{}
\end{aug}

\begin{abstract}
Two-step estimators often called upon to fit censored regression models in many areas of science and engineering. Since censoring incurs a bias in the naive least-squares fit, a two-step estimator first estimates the bias and then fits a corrected linear model. We develop a framework for performing valid \emph{post-correction inference} with two-step estimators. By exploiting recent results on post-selection inference, we obtain valid confidence intervals and significance tests for the fitted coefficients. 
\end{abstract}



\end{frontmatter}

\section{Introduction}
\label{sec:intro}

Censored regression models were developed to handle otherwise ordinary statistical problems in which the data is subject to censoring. Formally, censored models are ``two-part'' models that adds to a linear model, \eg{}
\[
y = X\beta + \epsilon,\,\epsilon\sim\cN(0,\sigma^2 I),
\]
a ``censoring'' or ``sample selection'' step prior to observation $y = g(y^*)$. Perhaps, the most common example of a censored regression model is the standard (Type I) Tobit model \citep{amemiya1985advanced}:
\BEQ
y = \max\{0,y^*\},\,y^*\sim\cN(X\beta,\sigma^2 I).
\label{eq:tobit-model}
\EEQ
% The censoring function $\max\{0,y^*\}$ has the form $\max\{0,y^*\} = \sum_\sigma I_{\reals_\sigma^n}(y^*)P_{\reals_+^n}(y^*)$, where $\reals_\sigma,\,\sigma = ({+,-,0},\dots,{+,-,0})$ are the orthants in $\reals^n$ and $P_{\reals_+^n}$ is the projector onto the nonnegative orthant. 
The step $\max\{0,y^*\}$ may be changed to $\max\{y_0,y^*\}$ without essentially changing the model by absorbing $y_0$ into the intercept term. Applications of this model are found throughout econometrics and the social sciences. Two examples are family income after receiving welfare \citep{keeley1978estimation} and the criminal behavior of habitual offenders \citep{witte1980estimating}. 

Most previous work on censored regression models focuses on estimation. Two competing approaches are the maximum likelihood estimator (MLE) and various two-step estimators. Under appropriate regularity conditions, both approaches produce asymptotically normal estimates. To perform inference, we call upon the standard battery of likelihood ratio, score/Lagrange multiplier, and Wald tests/intervals. Non-asymptotic inference is hard because the censoring mechanism incurs a bias in the usual least-squares estimate for $\beta$ that must be corrected for. Two-step estimators account for the censoring bias by first estimating a correction term and then fitting a (linear) regression model with the correction term. Since the correction term is random, the usual confidence intervals/significance tests for the regression coefficients fail to have the nominal coverage/error rates.

We focus on performing inference with two-step estimators conditioned on the outcome of the censoring mechanism. Since the correction term is completely determined by the outcome of the censoring mechanism, we call this form of inference \emph{post-correction inference}. Although we must condition on the censoring event to perform inference, the results are valid \emph{unconditionally}. We discuss the unconditional validity of Section \ref{sec:summary}.

\subsection{Background}

The first censored regression model was proposed by \cite{tobin1958estimation} to model household expenditure on durable goods. The model explicitly accounts for the fact that expenditure cannot be negative. Tobin called his model the model of \emph{limited dependent variables}. It and its generalizations are commonly known in econometrics as \emph{Tobit models}, a play on \emph{probit models} by \cite{goldberger1964econometric}. 
Censored regression models are also broadly applicable in many other areas of science and engineering. For example, the survival times of patients and the time to failure of a machine or a system are both censored by the length of the study. In their respective domains, these models are called \emph{survival} or \emph{duration} models. 

Since the 1970s, many generalizations of Tobin's original model has appeared and continue to appear in econometrics and the aforementioned areas. \cite{amemiya1985advanced} classified these models into five basic types according to the form of the likelihood function. We shall refer to various censored regression models according to Amemiya's classification. 
% \footnote{Amemiya's classification accounts for 95\% of Tobit models in econometrics.}

% A closely related group of models are \emph{truncated} regression models. In censored regression, only the response (endogenous variables) are unobserved, while in truncated regression, both the response and the predictors are unobserved.

Fitting a censored regression model is usually performed by M-estimation. To fit simple models, the maximum likelihood estimator (MLE) \citep{amemiya1973regression} (usually after a change of variable proposed by \cite{olsen1978note}) or Powell's least absolute deviations (LAD) estimator \citep{powell1984least} are preferred. The EM algorithm by \cite{dempster1977maximum} is also suitable for fitting censored regression models. For more complex models, two-step estimators \citep{heckman1976common} are more computationally efficient. Heckman originally proposed his two-step estimator for a Type 3 Tobit model, but readily adapts to other censored regression models. \cite{puhani2000heckman} compared two-step estimators with the MLE and concluded that two-step estimators are more robust when the covariates in the model are highly correlated. Otherwise, the MLE is more (statistically) efficient. 

Subject to appropriate regularity conditions, all the aforementioned estimators are asymptotically normal, thus inference is usually performed asymptotically. \cite{amemiya1973regression} showed the maximum likelihood estimate (MLE) is consistent and asymptotically normal (subject to homeoscedasticity and normality). The same is generally true of two-step estimators. \cite{powell1984least} showed his LAD estimator is consistent and asymptotically normal under more general conditions. This makes Powell's estimator especially attractive when the data is nonnormal or the errors are heteroscedastic. The main drawback to Powell's estimator is the computational expense of optimizing a nonsmooth function.

\subsection{Related work on post-selection inference}

This work is based on a framework for post-(model)-selection inference by \cite{lee2013post}. At the core of both frameworks are key distributional results on the supremum of a Gaussian process restricted to a convex set $\cK\subset\reals^p$ :
\[
\sup_{\eta\in\cK}\eta^T\epsilon,\,\epsilon\sim\cN(\mu,\Sigma).
\]
In this paper, we give a geometric derivation of a special case (when the set $\cK$ is polyhedral) of the main result (Theorem 1) in \cite{taylor2013tests} and refer to the source for a more general derivation. Like in \cite{taylor2013tests} and \cite{lee2013post}, normality is crucial to our results.

\subsection{Notation}

We follow the notation of \cite{amemiya1985advanced}. We are given $X\in\reals^{n\times p}$ and responses $y\in\reals^n$. We say $x_i^T$ for the $i$-th row of $X$ and $y_i$ for the $i$-th response. Often we need to distinguish between the vectors and matrices of censored and all observations; the former appear with an underbar, \eg{} $\y$ is the vector of uncensored responses. Usually, the uncensored responses are constrained to fall in the non-negative orthant $\reals_+^n$. Finally, $\phi(x)$ and $\Phi(x)$ are the pdf and CDF of the standard normal distribution.

\section{Post-correction inference with two-step estimators}
\label{sec:post-correction-inference}

\subsection{Two two-step estimators}

We describe two-step estimators in the context of the standard Tobit model. Since
\[
\Expect[y_i\mid y_i > 0] = x_i^T\beta + \Expect[\epsilon_i\mid\epsilon_i > x_i^T\beta]
\]
and the conditional expectation on the right side is usually nonzero (even when $\epsilon_i$ is not normal), performing (linear) regression on only the samples with positive responses produces biased estimates. When $\epsilon$ is normal, the bias term simplifies to
\BEQ\textstyle
\Expect[y_i\mid y_i > 0] = x_i^T\beta + \sigma\lambda\left(\frac1\sigma x_i^T\beta\right),
\label{eq:type-1-tobit-conditional-expect}
\EEQ
where $\lambda(x)$ is the inverse Mills ratio $\frac{\phi(x)}{1-\Phi(x)}$. A two-step estimator corrects the bias by first estimating the inverse Mills ratio $\lambda\left(\frac1\sigma x^T\beta\right)$ with probit regression, and then fitting a corrected linear model to the samples with positive responses. The corrected linear model is heteroscedastic:
\BEQ\textstyle
\var[\epsilon_i\mid\epsilon_i > x_i^T\beta] = \sigma^2 - \sigma^2x_i^T\alpha\lambda\left(\frac1\sigma x_i^T\beta\right) - \sigma^2\lambda\left(\frac1\sigma x_i^T\beta\right)^2,
\label{eq:type-1-tobit-conditional-variance}
\EEQ
so we can estimate $\beta$ more efficiently with weighted least squares in the second step. Since our focus is not estimation, we skip this topic and refer to \cite{amemiya1985advanced} for details.

\begin{algorithm}
\caption{Two-step estimator (Type 1 Tobit model)} \label{alg:two-step-estimator-standard-tobit}
\begin{algorithmic}[1]
\Require design matrix $X$, responses $y$
\State Estimate $\alpha = \frac1\sigma \beta$ with the probit MLE:
\Statex \pc\pc $\hat{\alpha} = \argmax_\alpha\,\prod_{i:y_i = 0}(1 - \Phi(x_i^T\alpha))\prod_{i:y_i > 0}\Phi(x_i^T\alpha)$. 
\State Regress the uncensored responses $\y$ on $\X$ and $\hat{\lambda} = \lambda(\X\hat{\alpha})$ to estimate $\beta$ and $\sigma$:
\Statex \pc\pc $\hat{\beta},\hat{\sigma} = \argmin_{\beta,\sigma}\frac{1}{2}\|\y - \sigma\hat{\lambda}- \X\beta\|_2^2$.
\end{algorithmic}
\end{algorithm}

\cite{heckman1979sample}, who proposed the first two-step estimator, interprets the first step in the two-step estimator can be interpreted as obtaining a correction for the missing bias term in \eqref{eq:type-1-tobit-conditional-expect}. As long as the probit MLE is consistent, the two-step estimator is also consistent. Further, when 
\BNUM
\item $\{x_i^T\}$ are uniformly bounded,
\item $\lim_{n\to\infty}X^TX/n$ (exists and) is positive definite,
\item $(\beta,\sigma)$ are in some (\emph{a priori known}) compact subset of $\reals^p$,
\ENUM
$\sqrt{n}(\hat{\beta} - \beta)$ is asymptotically normal \citep{amemiya1985advanced} and the usual battery of likelihood ratio, score, and Wald tests/intervals are valid. The two-step approach also generalizes readily to all five types of Tobit models in \cite{amemiya1985advanced}. We refer to \cite{amemiya1985advanced} for details.

In the second step, one can also regress all observations (including the censored observations) to obtain an estimate of $\beta$. This estimator is also consistent and asymptotically normal. Which estimator is more efficient? Unfortunately, the answer depends on the parameter $\beta$. However, the framework also applies to this estimator by considering a \emph{degenerate} constrained normal.

\subsection{A pivotal quantity for constrained normal variables}

Recall our goal is to perform inference on the fitted coefficients conditioned on the censoring event $\y > 0$. More generally, we are interested in targets of the form $\eta^T\y$ for some $\eta$. We begin by characterizing the conditional distribution of the response. Firstly, given $\y > 0$, $\y$ has a constrained normal distribution; \ie{} the conditional pdf of $\y$ is $f(\y\mid\y > 0) \propto \phi(\y)\ones_{\reals_+^{\n}}(\y)$. Secondly, given the censoring event, the correction term $\hat{\lambda}$ is simply
\[
\hat{\lambda} = \lambda(\X\hat{\alpha}),\,\hat{\alpha} = \argmax_\alpha\,\prod_{i\in S^c}(1 - \Phi(x_i^T\alpha))\prod_{i\in S}\Phi(x_i^T\alpha).
\]
The two-step estimate of $\gamma = (\beta,\sigma)$ is given by 
\[
\hat{\gamma} = \hat{Z}^\dagger\y = (\hat{Z}^T\hat{Z})^{-1}\hat{Z}^T\y,\,\hat{Z} = \BMAT \X & \hat{\lambda} \EMAT.
\]
To form (conditional) confidence intervals for or (conditionally) test the significance of the coefficients $\hat{\gamma}_j,\,j=1,\dots,p$, we study its distribution conditioned on the censoring event. In this section, we exploit distributional results in \cite{lee2013post} to derive a pivotal quantity.

\begin{theorem}
\label{thm:truncated-normal}
Let $y\sim\cN(\mu,\Sigma)$. Define $a = \frac{A\Sigma\eta}{\eta^T\Sigma\eta}$ and
\begin{align}
\V_\eta^+(y) &= \sup_{j:a_j < 0}\frac{1}{a_j}(b_j - (Ay)_j + a_j\eta^Ty) \label{eq:V-plus} \\
\V_\eta^-(y) &= \inf_{j:a_j > 0}\frac{1}{a_j}(b_j - (Ay)_j + a_j\eta^Ty) \label{eq:V-minus} \\
\V_\eta^0(y) &= \inf_{j:a_j = 0} b_j - (Ay)_j.
\label{eq:V-zero}
\end{align}
Then, $\eta^Ty$ conditioned on $Ay \le b$ and $(\V_\eta^+(y), \V_\eta^-(y),\V_\eta^0(y)) = (v^+,v^-,v^0)$, has a truncated normal distribution, \ie{}
\[\textstyle
\eta^Ty\mid Ay \le b,\V_\eta^+(y), \V_\eta^-(y),\V_\eta^0(y) \sim TN(\eta^T\mu,\eta^T\Sigma\eta,v^-,v^+).
\]
\end{theorem}

We defer a proof to Appendix \ref{sec:truncated-normal-proof} and focus on a geometric interpretation of Theorem \ref{thm:truncated-normal}. Let $C = \{y\in\reals^n\mid Ay \le b\}$. The functions $\V_\eta^+(y)$ and $\V_\eta^-(y)$ return the furthest we can move along $\tilde{\eta} = \Sigma^{1/2}\eta$  starting at $\tilde{y} = \Sigma^{-1/2}y$ while remaining in the set $\Sigma^{-1/2}C$, \ie
\begin{align*}
\V_\eta^+(y) &= \sup\{\lambda\in\reals:\lambda\tilde{\eta} + \tilde{y}\in\Sigma^{-1/2}C\} \\
\V_\eta^-(y) &= \inf\{\lambda\in\reals:\lambda\tilde{\eta} + \tilde{y}\in\Sigma^{-1/2}C\}.
\end{align*}
Since the set $\{\lambda:\lambda\tilde{\eta} + \tilde{y}\in\Sigma^{-1/2}C\}$ is equivalent to $\{\lambda:\lambda\Sigma^{1/2}\tilde{\eta} + \Sigma^{1/2}\tilde{y}\in C\}$, we have, the terms of the original variables, we have
\begin{align*}
\V_\eta^+(y) &= \sup\{\lambda:\lambda\Sigma\eta + y\in C\} \\
\V_\eta^-(y) &= \inf\{\lambda:\lambda\Sigma\eta + y\in C\}.
\end{align*}
The solution to both optimization problems are given by \eqref{eq:V-plus} and \eqref{eq:V-minus}. By conditioning on $(\V_\eta^+(y), \V_\eta^-(y), \V_\eta^0(y)) = (v^+,v^-,v_0)$, we are further restricting $y$ to fall on a slice of $C$ (the slice $\{\lambda\Sigma\eta + y\}\cap C$ parametrized by $\lambda$).  Theorem \ref{thm:truncated-normal} merely says $y$ restricted to this slice has a truncated normal distribution. This follows from the fact that $\V_\eta^+(y), \V_\eta^-(y), \V_\eta^0(y)$ are independent of $\eta^Ty$ by construction.

\begin{figure} 
%\vspace{6pc}
\includegraphics[width=0.5\textwidth]{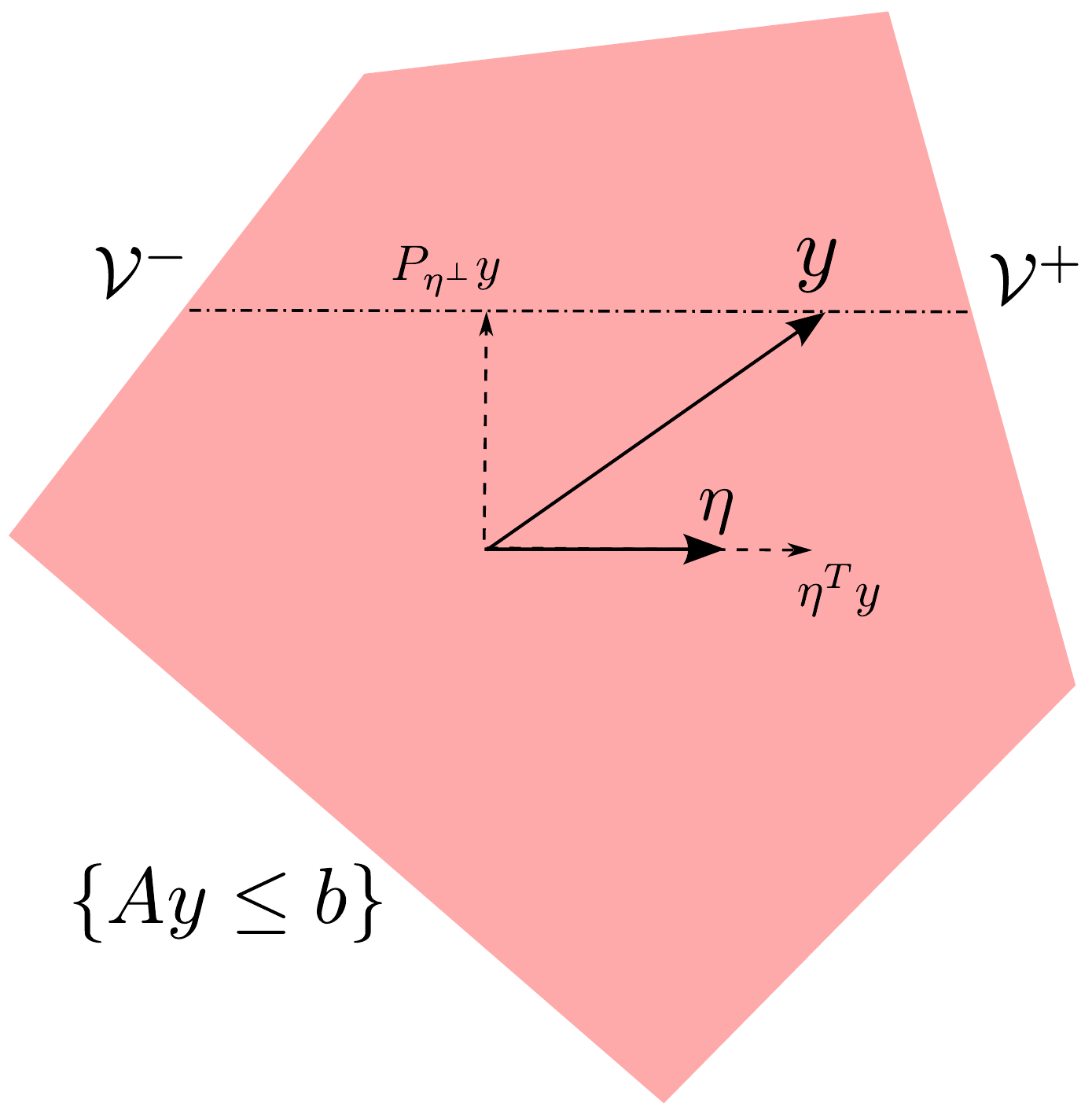}
\caption[]{Geometric intuition behind Theorem \ref{thm:truncated-normal}. To simplify the drawing, we assume $\Sigma = I$. By conditioning on $Ay \le b,\V^+(y),\V_\eta^-(y), \V^0(y) \ge 0$, we are restricting ourselves to the dotted slice of the shaded region. Thus conditioning on $\V_\eta^+(y), \V_\eta^+(y)$ is equivalent to conditioning on the component of $y$ orthogonal to $\eta$. Since the orthogonal component is independent of $\eta^Ty$, conditioning on it does not affect the normality of $\eta^Ty$.}
\label{fig:slice-orthant}
\end{figure}

Given $y\sim\cN(\mu,\Sigma)$, we restrict to the slice $\{\lambda\Sigma\eta + y\}\cap\reals_+^n$. To obtain a $\unif(0,1)$ distributed pivotal quantity, we apply a CDF transform. Figure \ref{fig:pval-cdf} shows results from a simulation that empirically confirm the pivotal quantity is uniformly distributed.

\begin{corollary}
\label{cor:pivotal-quantity}
Let $y\sim\cN(\mu,\Sigma)$. Define 
\BEQ
F(x,\mu,\sigma^2,a,b) = \frac{\Phi\left(\frac{x-\mu}{\sigma}\right) - \Phi\left(\frac{v^--\mu}{\sigma}\right)}{\Phi\left(\frac{v^+-\mu}{\sigma}\right) - \Phi\left(\frac{v^--\mu}{\sigma}\right)}.
\label{eq:probability-integral-transform}
\EEQ
Then $F(\eta^Ty,\eta^T\mu,\eta^T\Sigma\eta,\V_\eta^-(y),\V_\eta^+(y))$ is a (conditional) pivotal quantity with a $\unif(0,1)$ distribution, \ie
\BEQ
F(\eta^Ty,\eta^T\mu,\eta^T\Sigma\eta,\V_\eta^-(y),\V_\eta^+(y))\mid Ay\le b\sim\unif(0,1)
\label{eq:pivotal-quantity-1}
\EEQ
where $\V_\eta^+(y),\,\V_\eta^-(y),\,\V_\eta^0(y)$ are given by \eqref{eq:V-plus}, \eqref{eq:V-minus}, \eqref{eq:V-zero}.
\end{corollary}

\begin{proof}
By Theorem \ref{thm:truncated-normal}, we know 
\[
\eta^Ty\mid y > 0,\V_\eta^+(y), \V_\eta^-(y),\V_\eta^0(y),
\]
where $\V_\eta^+(y), \V_\eta^-(y),\V_\eta^0(y)$ are given by \eqref{eq:V-plus}, \eqref{eq:V-minus}, \eqref{eq:V-zero}, has a truncated normal distribution. We apply the CDF transform to deduce
\[
F(\eta^Ty,\eta^T\mu,\eta^T\Sigma\eta,\V_\eta^-(y),\V_\eta^+(y))\mid y > 0,\V_\eta^+(y), \V_\eta^-(y), \V_\eta^0(y)
\]
is uniformly distributed. Since this holds for any $\V_\eta^+(y),\,\V_\eta^-(y),\,\V_\eta^0(y)$, we deduce
\[
F(\eta^Ty,\eta^T\mu,\eta^T\Sigma\eta,\V_\eta^-(y),\V_\eta^+(y))\mid y > 0.
\]
is also uniformly distributed.
\end{proof}

\begin{figure}
\centering
\includegraphics[width=0.5\textwidth]{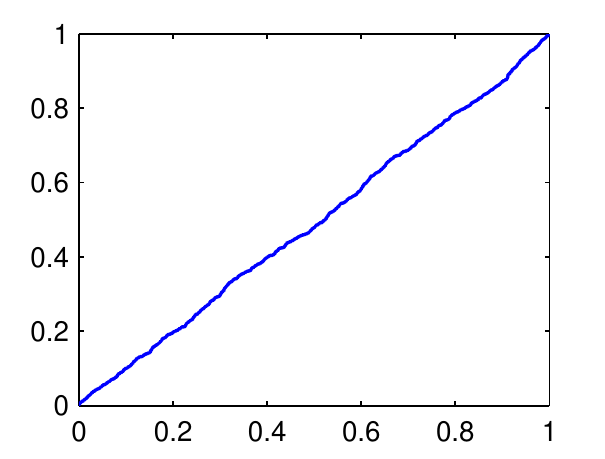}
% note that files may not be rotated
\caption{Empirical CDF of 10000 samples of the pivotal quantity \eqref{eq:pivotal-quantity-1} obtained by rejection sampling. The empirical distribution is very close to that of a $\unif(0,1)$ variable.}
\label{fig:pval-cdf}
\end{figure}

\subsection{Confidence intervals for the fitted coefficients}

Recall we seek valid confidence intervals for $\beta_j$ \emph{conditioned} on the censoring event $\y > 0$. Let $\eta_j = (\X^\dagger)^T e_j$. 
By Corollary \ref{cor:pivotal-quantity}, we have
\[
F(\eta_j^T\y,\eta_j^T\X\beta,\sigma^2\norm{\eta_j}_2^2,\V_{\eta_j}^-(\y),\V_{\eta_j}^+(\y))\mid\y > 0\sim\unif(0,1),
\]
where $\V_\eta^+(y),\,\V_\eta^-(y),\,\V_\eta^0(y)$ are given by \eqref{eq:V-plus}, \eqref{eq:V-minus}, \eqref{eq:V-zero} with $A = -I$ and $b = 0$. The second argument to $F$ simplifies to 
\[
\eta_j^T\X\beta = e_j^T\X^\dagger\X\beta = \beta_j.
\]
%The matrix $\D$ is diagonal with the plug-in estimates of $\var[\y_i\mid\y_i > 0]$ on the diagonal:
%\[
%d_{ii} = \hat{\sigma}^2 - \hat{\sigma}^2x_i^T\hat{\alpha}\lambda\left(x_i^T\hat{\alpha}\right) - \hat{\sigma}^2\lambda\left(x_i^T\hat{\alpha}\right)^2
%\]
To obtain valid confidence intervals for $\beta_j$, we simply ``invert'' the pivotal quantity. The set
\BEQ\textstyle
\left\{\nu\in\reals:\frac{\alpha}{2} \le F(\eta_j^T\y,\nu,\sigma^2\norm{\eta_j}_2^2,\V_{\eta_j}^-(\y),\V_{\eta_j}^+(\y)) \le 1 - \frac{\alpha}{2}\right\}
\label{eq:invert-pivot}
\EEQ
is a $1-\alpha$ confidence interval for $\beta_j$. The endpoints $\frac{\alpha}{2}$ and $1 - \frac{\alpha}{2}$ were chosen arbitrarily. By Lemma \ref{lem:F-monotone}, $F$ decreases monotonically in $\nu$. To obtain an intervals, we need to solve two univariate root-finding problem. Figure \ref{fig:corr-intervals} shows results from two simulations that compare the coverage of the corrected intervals versus the normal intervals. 

\begin{lemma}
\label{lem:intervals}
Let $\eta_j = (\X^\dagger)^T e_j$. Define $\nu_\alpha(\y)$ to be the (unique) root of 
\[
F(\eta_j^T\y,\nu,\sigma^2\norm{\eta_j}_2^2,\V_{\eta_j}^-(\y),\V_{\eta_j}^+(\y)) = \alpha.
\]
Given $\y > 0$, $[\nu_{\alpha/2}(\y),\nu_{1-\alpha/2}(\y)]$ is a valid $1-\alpha$ confidence interval for $\beta_j$:
\[
\Prob(\beta_j\in[\nu_{\alpha/2}(\y),\nu_{1-\alpha/2}(\y)]\mid\y>0) = 1 - \alpha.
\]
\end{lemma}

\begin{figure}
\includegraphics[width = .48\textwidth]{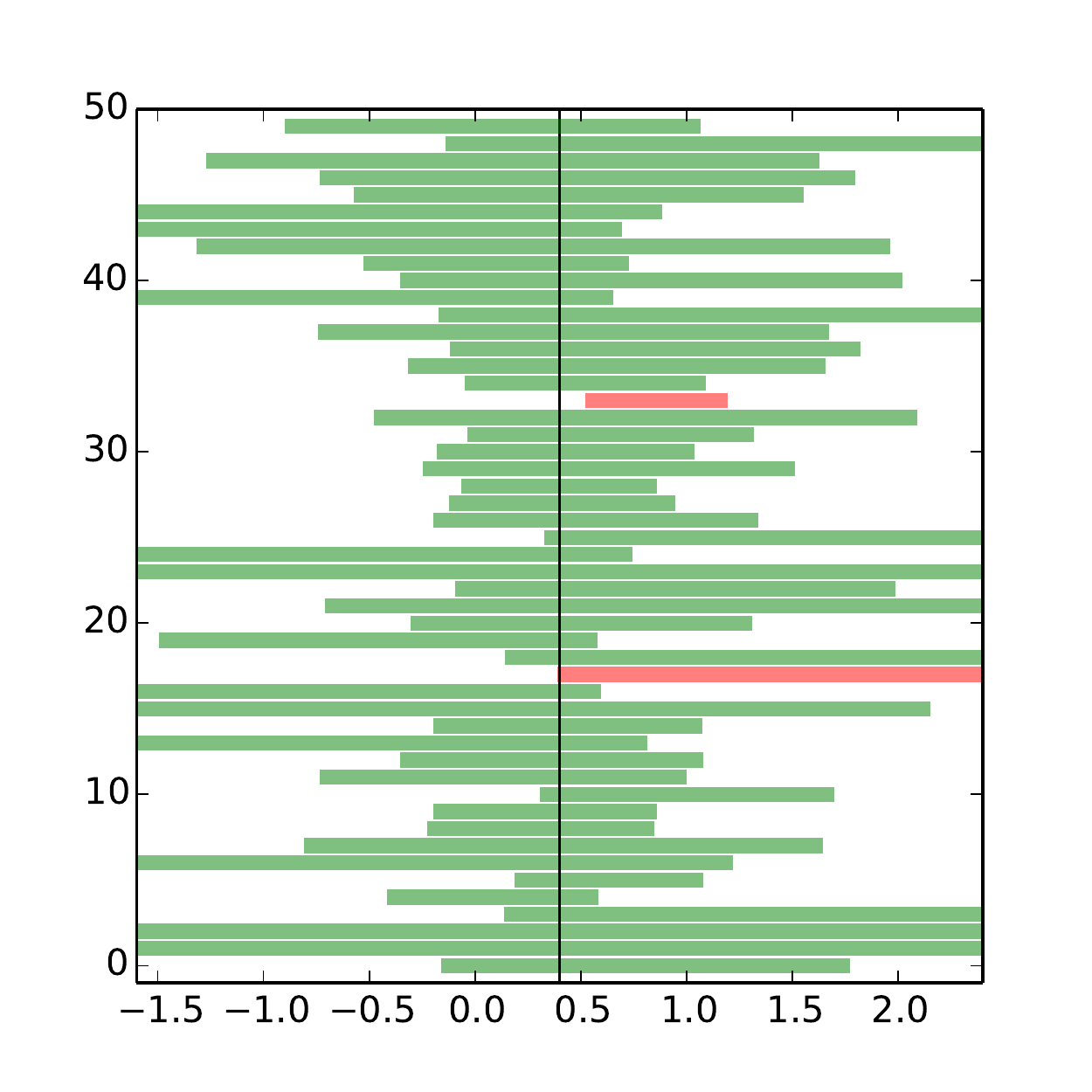}
\includegraphics[width = .48\textwidth]{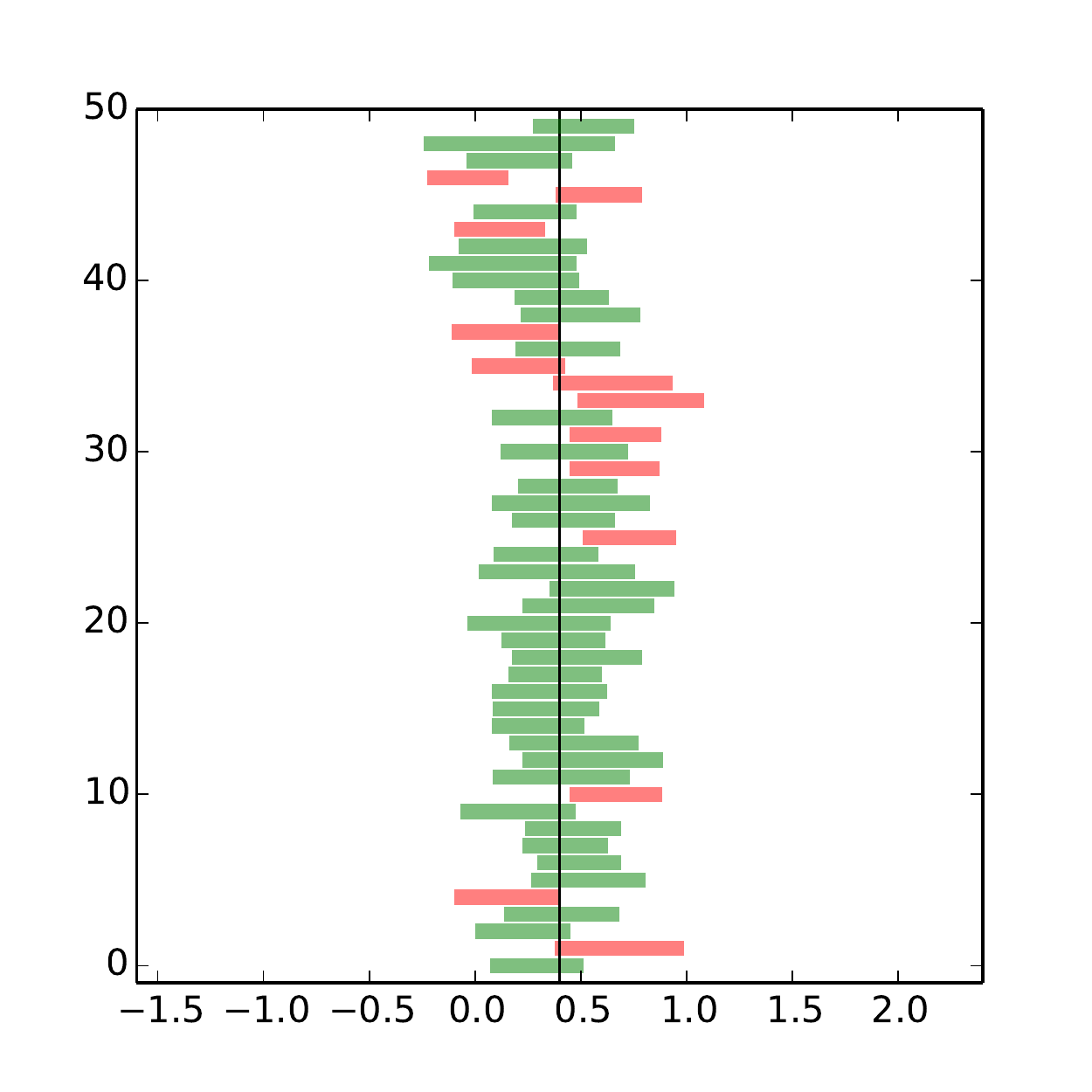}
\caption{95\% corrected (left) and normal (right) intervals for $\beta_1$ in a Type 1 Tobit model over 50 simulated data sets (n = 100, p = 10, $\sigma^2$ = 1.0, standard Gaussian design). A green bar means the confidence interval covers $\beta_1$ while a red bar means otherwise. Although the corrected intervals are wider, they cover $\beta_1$ (black line) at the nominal rate. Empirically, we observe the coverage of the normal intervals becomes worse as the signal strength decreases.}
\label{fig:corr-intervals}
\end{figure}

The strength of the aforementioned results do come at a price. In this case, the price is wider confidence intervals. For example, suppose $\eta^T\y$ falls near an endpoint of the truncation interval $[v^-,v^+]$, say $v^+$. Since this is likely to occur when $\eta^T\mubar$ is very large, there is a wide range of large values of $\eta^T\mubar$ that make this observation likely. However, when $\eta^T\y$ is not near an endpoint of the truncation interval (and the truncation interval) is large, then the confidence interval will be comparable to the least-squares (normal) intervals. Figure \ref{fig:tobit-1-interval-width} compares the coverage of the truncated intervals with that of normal intervals. 

\begin{figure}
\includegraphics[width=0.48\textwidth]{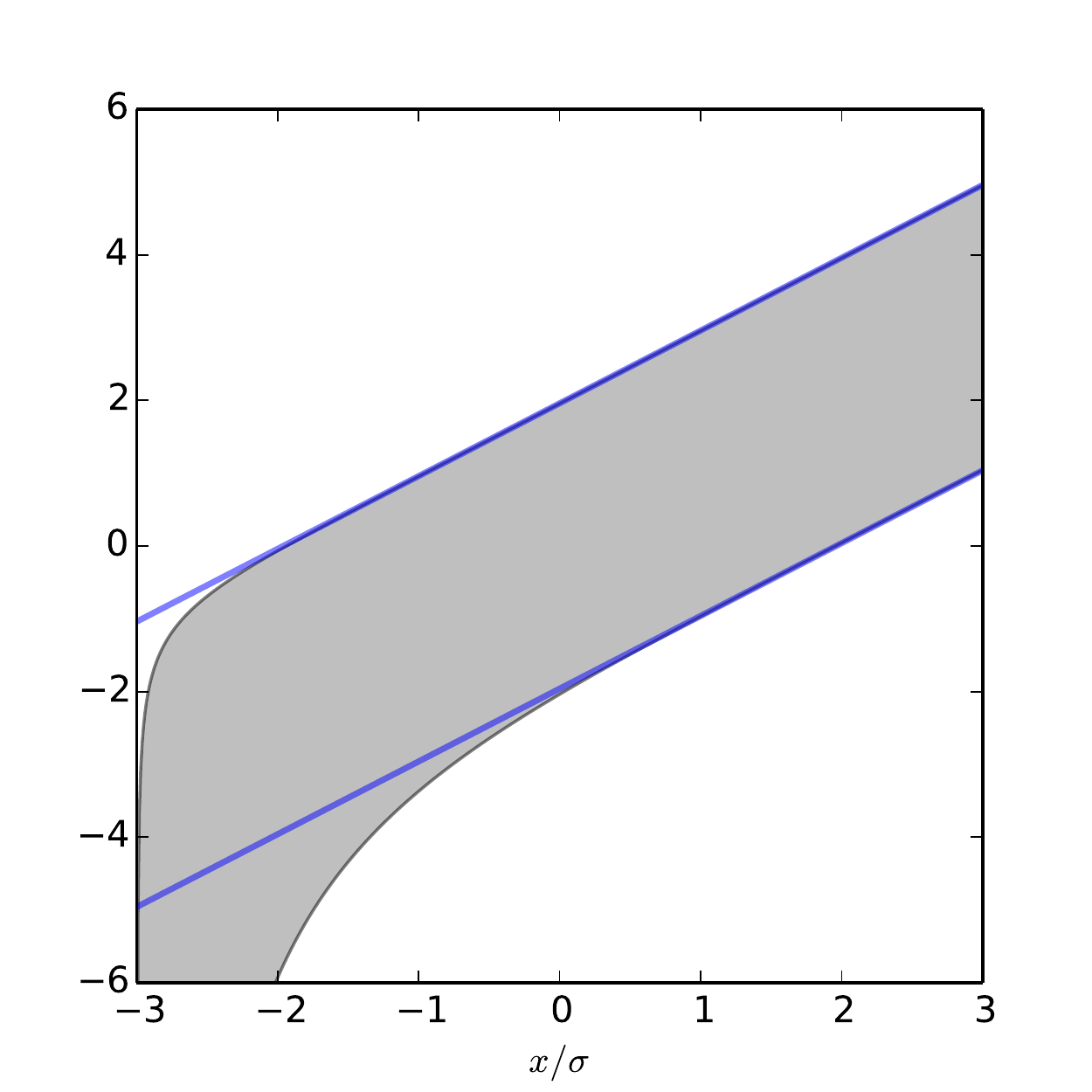}
\includegraphics[width=0.48\textwidth]{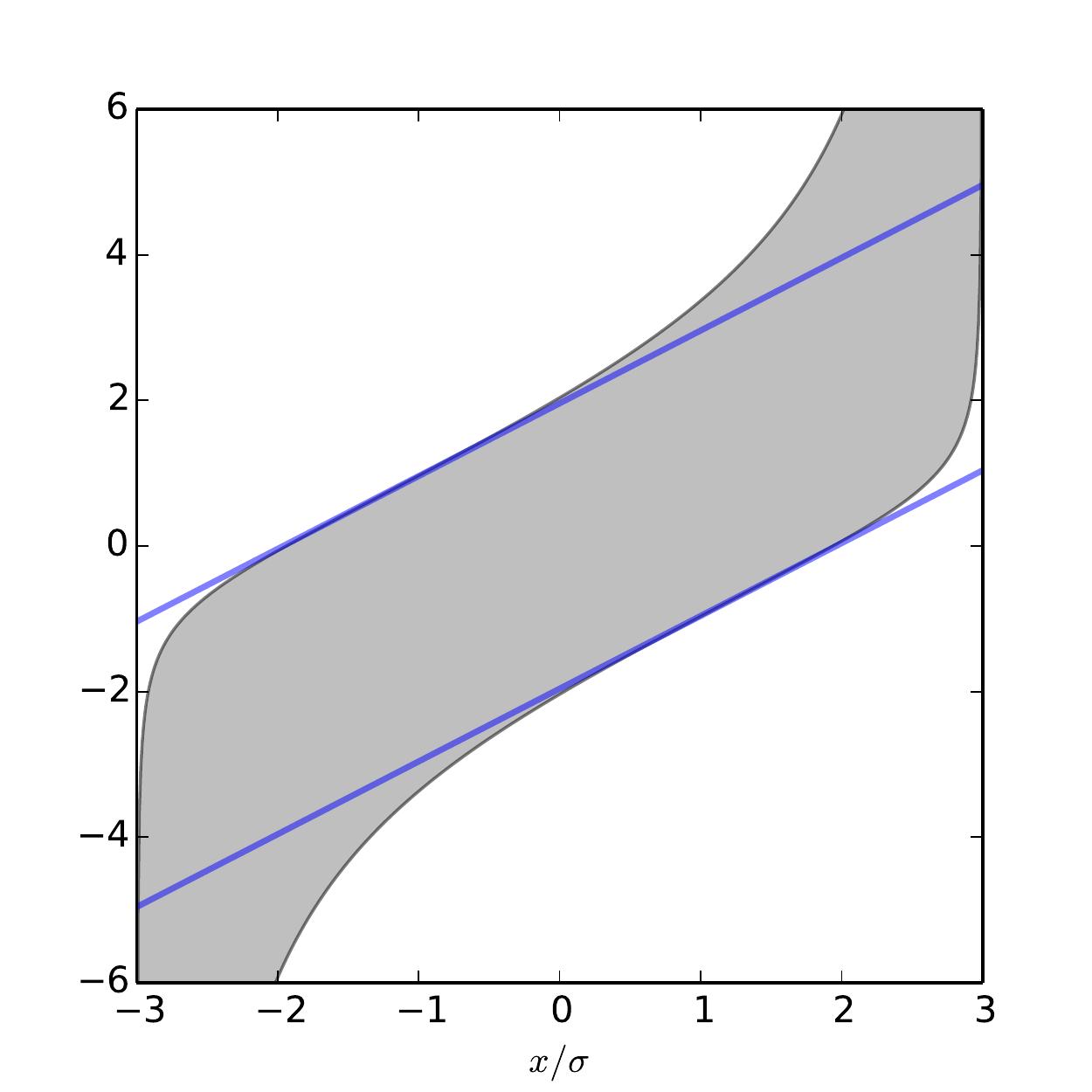}
% note that files may not be rotated
\caption{95\% truncated (gray region) and normal (area between blue lines) confidence intervals as a function of $\frac{x}{\sigma}$. The truncation regions are $[-3\sigma,\infty]$ and $[-3\sigma,3\sigma]$ in the figures on the left and right. When $x$ falls near the boundary of the truncation region (in terms of $\sigma$), the conditional interval is much wider than the normal interval.}
\label{fig:tobit-1-interval-width}
\end{figure}

\subsection{Testing the significance of fitted coefficients}

We also seek to test the significance of the $j$-th regression coefficient $\hat{\beta}_j$, \ie{} test the hypothesis 
\[
H_0:\eta_j^T\X\beta = \beta_j = 0,\,\eta_j =(\X^\dagger)^Te_j.
\]
The output of the censoring mechanism is random, so $H_0$ is a \emph{random} hypothesis.
One possible interpretation of testing $H_0$ is we are testing $H_0$ conditioned on the censoring event. Since $H_0$ is fixed given $\y > 0$, this is a hypothesis in the conventional sense. As we shall see, there is also an \emph{unconditional} interpretation of testing $H_0$.

Under $H_0$, we know (by Corollary \ref{cor:pivotal-quantity}) 
\[
F(\eta_j^T\y,0,\sigma^2\norm{\eta_j}_2^2,\V_{\eta_j}^-(\y),\V_{\eta_j}^+(\y))\mid\y > 0\sim\unif(0,1),
\]
We reject $H_0$ when the pivotal quantity falls in the top or bottom $\frac\alpha2$ quantile to obtain a valid $\alpha$-level test for $H_0$. This test is the post-correction counterpart to the usual t-test and controls Type I error conditioned on the censoring event:
\[
\Prob(\text{reject }H_0\mid\y > 0,H_0) \le \alpha.
\]
Since the test controls the Type I error rate at 5\% for all possible outcomes of the censoring mechanism, the test also controls Type I error \emph{unconditionally}. In other words, the test also is a valid unconditional test of $H_0$.

\begin{lemma}
\label{lem:tests}
Let $\eta_j = (\X^\dagger)^T e_j$. Given $\y > 0$, the test that rejects when 
\[\textstyle
F(\eta_j^T\y,0,\sigma^2\norm{\eta_j}_2^2,\V_{\eta_j}^-(\y),\V_{\eta_j}^+(\y)) \in \left[0,\frac{\alpha}{2}\right)\cup\left(\frac{\alpha}{2},1\right]
\]
is a valid $\alpha$-level test for $H_0:\beta_j = 0$.
\end{lemma}

%\section{Example: Stanford heart transplant data}
%
%The Stanford heart transplant data contains information on the survival times of 184 heart transplant patients from October 1967 to February 1980. Their survival times, uncensored or censored, along with their ages at the time of their first transplant and T5 mismatch scores comprise the data. The mismatch score measures the genetic compatibility between the donor heart and the recipient. 
%
%For 27 of the 184 patients, the T5 mismatch scores are missing because tissue typing results were unavailable. Our analysis only includes the remaining 157 patients with complete records. Of these 157 patients, the survival times of 55 were censored (\ie{} they were alive in February 1980) and the remaining 102 were uncensored. The fitted coefficients, along with confidence intervals and significance, are given in Table \ref{tab:stanford-lm}.

\section{Handling other censored regression models}

Two-step estimators are broadly applicable to many censored regression models, and the post-correction inference framework also generalizes accordingly. We already described how to perform post-correction inference for the standard (Type 1) Tobit model in Section \ref{sec:post-correction-inference}. In this section, we outline how the framework applies to other censored regression models.

\subsection{Accelerated failure time models}

Accelerated failure time (AFT) models are often used to model survival or duration data in science and engineering. Mathematically, AFT models are very similar to Tobit models. Let $t_i,\,i=1,\dots,n$ be the (possibly censored) failure times of $n$ units. AFT models posit a linear relationship between $\log(t_i)$ and some variables $x_i^T$; \ie
\BEQ
\log(t_i) = x_i^T\beta + \epsilon_i,\,i=1,\dots,n
\label{eq:accelerated-failure-time}
\EEQ
where $\epsilon$ is independent error. The distribution of $\epsilon$ determines AFT model in use. Since we only observe $t_i$ if unit $i$ fails before the end of the testing period (at some known $T$), the failure times are right-censored at $T$.
%\ie we only observe $t_i$ if unit $i$ fails before some time $T$.  
If the error is normally distributed (a log-normal AFT model), then we recognize the standard Tobit model with right-censoring at $\log(T)$:
\BEQ
y \sim \min\{y^*,\log(T)\},\,y^*\sim\cN(X\beta,\sigma^2I).
\label{eq:accelerated-failure-time-tobit}
\EEQ

To handle right censoring at $\log(T)$, we absorb $\log(T)$ into the intercept term and take $-\log(t_i)$ as the response. If the testing period is different for each unit, the model is slightly changed because the resulting model is equivalent to \eqref{eq:accelerated-failure-time-tobit} where an entry of $\beta$ is known. To fit a log-normal AFT model, we simply take $\log(T_i) - \log(t_i)$ as the response and call upon Algorithm \ref{alg:two-step-estimator-standard-tobit}. To obtain confidence intervals or test the significance of the fitted coefficients, we rely on the pivotal quantity \eqref{eq:pivotal-quantity-1}.

\subsection{Type 3 Tobit model}

The Type 3 model is:
\BEQ
\begin{aligned}
y_1^* &= X_1\beta_1 + \epsilon_1,\,\epsilon_1\sim\cN(0,\sigma_1^2I) \\
y_2^* &= X_2\beta_2 + \epsilon_2,\,\epsilon_1\sim\cN(0,\sigma_2^2I) \\
y_{1,i} &= \BCAS y_{1,i}^* & \text{if }y_{1,i}^* > 0 \\ 0 & \text{otherwise} \ECAS,\,i=1,\dots,n \\
y_{2,i} &= \BCAS y_{2,i}^* & \text{if }y_{1,i}^* > 0 \\ 0 & \text{otherwise} \ECAS,\,i=1,\dots,n.
\end{aligned}
\label{eq:type-3-tobit-model}
\EEQ
The error terms $\epsilon_1$ and $\epsilon_2$ are again normal with covariance $\sigma_{12}^2I$. The Type 1 \eqref{eq:tobit-model} and Type 2 \eqref{eq:type-2-tobit-model} Tobit models are special cases of the Type 3. 
Type 2 differs only in that $y_1$ is not observed, and Type 1 is special case of Type 2 in which $X_1 = X_2$ and $\epsilon_1 = \epsilon_2$. 
% However, unlike the Type 1 Tobit, the uncensored responses $\y_2$ may take negative values: $y_{1,i}^* > 0$ merely implies $y_{2,i} \ne 0$. 

To derive the two-step estimator for the Type 3 Tobit, we first evaluate the bias incurred by censoring. For $y_1$, the bias has the same form as \eqref{eq:type-1-tobit-conditional-expect}:
\begin{align*}
\Expect[y_{1,i}\mid y_{1,i}^* > 0] &= x_{1,i}^T\beta_1 + \Expect[\epsilon_{1,i}\mid\epsilon_{1,i} > -x_{1,i}^T\beta_1] \\
&\textstyle= x_{1,i}^T\beta_1 + \sigma_1\lambda\left(\frac{1}{\sigma_1}x_{1,t}^T\beta_1\right),
\end{align*}
where $\lambda(x)$ is the inverse Mills ratio. For $y_2$, the bias is similar:
\[
\Expect[y_{2,i}^*\mid y_{1,i}^* > 0] = x_{2,i}^T\beta_2 + \Expect[\epsilon_{2,i}\mid\epsilon_{1,i} > -x_{1,i}^T\beta_1].
\]
Since $\epsilon_1$ and $\epsilon_2$ are jointly normal, $\Expect[\epsilon_{2,i}\mid\epsilon_{1,i} > -x_{1,i}^T\beta_1]$ is a simple linear function of $\Expect[\epsilon_{1,i}\mid\epsilon_{1,i} > -x_{1,i}^T\beta_1]$:
\BEQ\textstyle
\Expect[y_{2,i}\mid y_{1,i}^* > 0] = x_{2,i}^T\beta_2 + \tau\lambda\left(\frac{1}{\sigma_1}x_{1,i}^T\beta_1\right),\,\tau = \frac{\sigma_{12}}{\sigma_1}
\label{eq:conditional-expectation-z2}
\EEQ
The two-step estimator first estimates $\lambda\left(\frac{1}{\sigma_1}x_{1,i}^T\beta_1\right)$ with probit regression and then fits a corrected linear model to uncensored observations. As long as the probit MLE is consistent, the two-step estimator produces consistent estimates of $\beta_1$ and $\beta_2$.\footnote{\cite{olsen1980least} notes that the consistency of the two-step estimator does not require the joint normality of $\epsilon_1$ and $\epsilon_2$ provided $\epsilon_1$ is normal and $y_{2,i}^* = x_{2,i}^T\beta_2 + \frac{\sigma_{12}}{\sigma_1^2}(y_{1,i}^* - x_{i,1}^T\beta_1) + \epsilon'_{2,i}$ for some $\epsilon'_{2,i}$ independent of $y_{1,i}^*$. The asymptotic variance of two-step estimators under these more general conditions is given by \cite{lee1982some}.} The corrected linear model is again heteroscedastic:
\BEQ
\var[\epsilon_{2,i}\mid\epsilon_{1,i} > -x_{1,i}^T\beta_1] = \sigma_2^2 - \tau^2(x_{1,i}^T\alpha_1\lambda(x_{1,i}^T\alpha_1) + \lambda(x_{1,i}^T\alpha_1)^2).
\label{eq:type-3-tobit-conditional-variance}
\EEQ

\begin{algorithm}
\caption{Heckman's two-step estimator (Type 3 Tobit model)} 
\label{alg:two-step-estimator-type-3-tobit}
\begin{algorithmic}[1]
\Require design matrix $X$, responses $y$
\State Estimate $\alpha_1 = \frac{1}{\sigma_1} \beta_1$ with the probit MLE: 
\Statex \pc\pc $\hat{\alpha}_1 = \argmax_\alpha\,\prod_{i:y_i = 0}(1 - \Phi(x_{1,i}^T\alpha))\prod_{i:y_i \ne 0}\Phi(x_{1,i}^T\alpha)$. 
\State Regress the uncensored responses $\y_1$ on $\X_1$ and $\lambda(\X_1\hat{\alpha}_1)$ to estimate $\beta_1$: 
\Statex \pc\pc $\hat{\beta}_1,\hat{\sigma}_1 = \argmin_{\beta,\sigma}\frac12\|\y_1 - \sigma_1\hat{\lambda}- \X_1\beta\|_2^2$. 
\State Regress the uncensored responses $\y_2$ on $\X_2$ and $\hat{\lambda} = \lambda(\X_2\hat{\alpha}_1)$ to estimate $\beta_2$: 
\Statex \pc\pc $\hat{\beta}_2,\hat{\tau} = \argmin_{\beta,\tau}\frac12\|\y_2 - \tau\hat{\lambda}- \X_2\beta\|_2^2$. 
\State Obtain an estimate of $\sigma_2$ from the residuals of step 3:
\Statex \pc\pc $\hat{\sigma}_2 = \frac{1}{n-p}\|\y_2 - \X_2\hat{\beta}_2\|_2^2 + \hat{\tau}^2\frac1n\sum_{i=1}^nx_{1,i}^T\hat{\alpha}_1\lambda(x_{1,i}^T\hat{\alpha}_1) + \lambda(x_{1,i}^T\hat{\alpha}_1)^2$.
\end{algorithmic}
\end{algorithm}

%We seek to perform inference on the fitted coefficients $\hat{\beta}_1$ and $\hat{\beta}_2$ \emph{conditioned} on $\y_1 > 0$. In the Type 3 Tobit model, the pdf of $\y_1$ given $\y_1 > 0$ is $f(\y_1) \propto \phi(\y_1)\ones_{\reals_+^{\n}}(\y_1)$ and that of $\y_2$ is simply $\phi(\y_2)$. Given the censoring event, the correction term $\hat{\lambda}$, and thus $\hat{Z}$, is fixed:
%\[
%\hat{\lambda} = \lambda(\X\hat{\alpha}),\,\hat{\alpha} = \argmax_\alpha\,\prod_{i:y_i = 0}(1 - \Phi(x_{1,i}^T\alpha))\prod_{i:y_i \ne 0}\Phi(x_{1,i}^T\alpha).
%\]
 The two-step estimate of $\beta_1$ is given by the first part of
\[
\hat{\gamma}_1 = \hat{Z}_1^\dagger\y_1,\,\hat{Z}_1 = \BMAT \X_1 & \hat{\lambda} \EMAT,
\]
Given the censoring event $\y_1 > 0$, $\y_1$ has a constrained normal distribution. To perform (valid) post-correction inference for $\beta_{1,j} = e_j^T\X_1^\dagger\X_1\beta_1$, we apply the results in Section \ref{sec:post-correction-inference}. Let $\eta_{1,j} = (\X_1^\dagger)^Te_j$. By Corollary \ref{cor:pivotal-quantity}, we know
\[
F(\eta_{1,j}^T\y_1, \beta_{1,j},\sigma_1^2\norm{\eta_{1,j}}_2^2, \V_{\eta_{1,j}}^-(\y_1), \V_{\eta_{1,j}}^+(\y_1))\mid\y_1 > 0 \sim \unif(0,1),
\]
where $\V_\eta^+(y),\,\V_\eta^-(y),\,\V_\eta^0(y)$ are given by \eqref{eq:V-plus}, \eqref{eq:V-minus}, \eqref{eq:V-zero} with $A = -I$ and $b = 0$. To form intervals for $\beta_{1,j}$, we ``invert'' $F$ to obtain 
\[
\Prob(\beta_{1,j}\in[\nu_{\alpha/2}(\y_1),\nu_{1-\alpha/2}(\y_1)]\mid\y_1>0) = 1 - \alpha,
\]
where $\nu_\alpha(\y)$ is the (unique) root of 
\BEQ
F(\eta_{1,j}^T\y_1, \nu,\sigma_1^2\norm{\eta_{1,j}}_2^2, \V_{\eta_{1,j}}^-(\y_1), \V_{\eta_{1,j}}^+(\y_1)) = \alpha.
\label{eq:F-roots}
\EEQ
To test the significance of $\beta_{1,j}$, we form a p-value
\[
F(\eta_{1,j}^T\y_1, 0,\sigma_1^2\norm{\eta_{1,j}}_2^2, \V_{\eta_{1,j}}^-(\y_1), \V_{\eta_{1,j}}^+(\y_1))
\]
with $\unif(0,1)$ distribution under the null $H_0:\beta_j = 0$. We reject $H_0$ when the p-value is smaller than $\frac{\alpha}{2}$ or larger than $1-\frac{\alpha}{2}$. The corrected intervals for $\beta_1$ should behave comparably with the corrected intervals for $\beta$ in a Type 1 Tobit model. 

The two-step estimate of $\beta_2$ is given by the first part of 
\[
\hat{\gamma}_2 = \hat{Z}_2^\dagger\y_2,\,\hat{Z}_2 = \BMAT \X_2 & \hat{\lambda} \EMAT,
\]
To perform post-correction inference on $\beta_2$, we must account for the dependence between $\y_1$ and $\y_2$. Given the censoring event $\y_1 > 0$, the pair $(\y_1,\y_2)$ has a constrained normal distribution, \ie{} 
\[
\BMAT \y_1 \\ \y_2\EMAT \sim \cN\left(\BMAT X_1\beta_1 \\ X_2\beta_2\EMAT,\BMAT \sigma_1^2 I & \sigma_{12} I \\ \sigma_{12} I & \sigma_2^2 I\EMAT\right),\text{ subject to }\y_1 > 0.
\]
Let $\y = (\y_1,\y_2)$ and $\mubar,\Sigmabar$ be its (unconditional) expected value and covariance. To form intervals for $\beta_{2,j}$, we first express our target as
\[
\beta_{2,j} = e_j^T\X_2^\dagger\X_2\beta_2 = \underbrace{e_j^T\BMAT 0 & \X_2^\dagger\EMAT}_{\eta_{2,j}^T}\mubar.
\]
By Corollary \ref{cor:pivotal-quantity}, we know
\BEQ
F\left(\eta_{2,j}^T\y, \beta_{2,j},\eta_{2,j}^T\Sigmabar\eta_{2,j}, \V_{\eta_{2,j}}^-(\y), \V_{\eta_{2,j}}^+(\y)\right)\mid\y_1 > 0,
\label{eq:pivotal-quantity-3}
\EEQ
where $\V_\eta^+(y),\,\V_\eta^-(y),\,\V_\eta^0(y)$ are given by \eqref{eq:V-plus}, \eqref{eq:V-minus}, \eqref{eq:V-zero}, with $A = \BMAT-I & 0\EMAT$ and $b = 0$, is uniformly distributed. With this pivotal quantity, we derive confidence intervals and significance tests for  $\beta_{2,j}$ like in Section \ref{sec:post-correction-inference}.

To form confidence intervals for $\beta_{2,j}$, we ``invert'' $F$ to obtain 
\[
\Prob(\beta_{2,j}\in[\nu_{\alpha/2}(\y),\nu_{1-\alpha/2}(\y)]\mid\y_1>0) = 1 - \alpha,
\]
where $\nu_\alpha(\y)$ is the (unique) root of 
\[
F(\eta_{2,j}^T\y, \nu,\eta_{2,j}^T\Sigmabar\eta_{2,j}, \V_{\eta_{2,j}}^-(\y), -\infty, \infty) = \alpha.
\]
To test the significance of $\beta_{2,j}$, we form a p-value
\[
F(\eta_{2,j}^T\y, 0,\eta_{2,j}^T\Sigmabar\eta_{2,j}, -\infty, \infty)
\]
with $\unif(0,1)$ distribution under the null $H_0:\beta_{2,j} = 0$. We reject $H_0$ when the p-value is smaller than $\frac{\alpha}{2}$ or larger than $1-\frac{\alpha}{2}$. Figure \ref{fig:tobit-3-intervals} shows result from two simulations that compare the coverage of the corrected intervals versus the normal intervals. We summarize our results in a pair of lemmas. 

\begin{lemma}
\label{lem:tobit-3-intervals}
Let $\eta_{2,j} = \BMAT 0 & \X_2^\dagger\EMAT^Te_j$. Define $\nu_\alpha(\y)$ to be the (unique) root of 
\[
F\left(\eta_{2,j}^T\y, \beta_{2,j},\eta_{2,j}^T\Sigmabar\eta_{2,j}, \V_{\eta_{2,j}}^-(\y), \V_{\eta_{2,j}}^+(\y)\right) = \alpha.
\]
Given $\y > 0$, $[\nu_{\alpha/2}(\y),\nu_{1-\alpha/2}(\y)]$ is a valid $1-\alpha$ confidence interval for $\beta_{2,j}$:
\[
\Prob(\beta_{2,j}\in[\nu_{\alpha/2}(\y),\nu_{1-\alpha/2}(\y)]\mid\y_1>0) = 1 - \alpha,
\]
\end{lemma}

\begin{lemma}
\label{lem:tobit-3-tests}
Let $\eta_{1,j} = \BMAT 0 & \X_2^\dagger\EMAT^Te_j$. Given $\y > 0$, the test that rejects when 
\[\textstyle
F\left(\eta_{2,j}^T\y, 0,\eta_{2,j}^T\Sigmabar\eta_{2,j}, \V_{\eta_{2,j}}^-(\y), \V_{\eta_{2,j}}^+(\y)\right) \in \left[0,\frac{\alpha}{2}\right)\cup\left(\frac{\alpha}{2},1\right]
\]
is a valid $\alpha$-level test for $H_0:\beta_{2,j} = 0$.
\end{lemma}

\begin{figure}
\includegraphics[width = .48\textwidth]{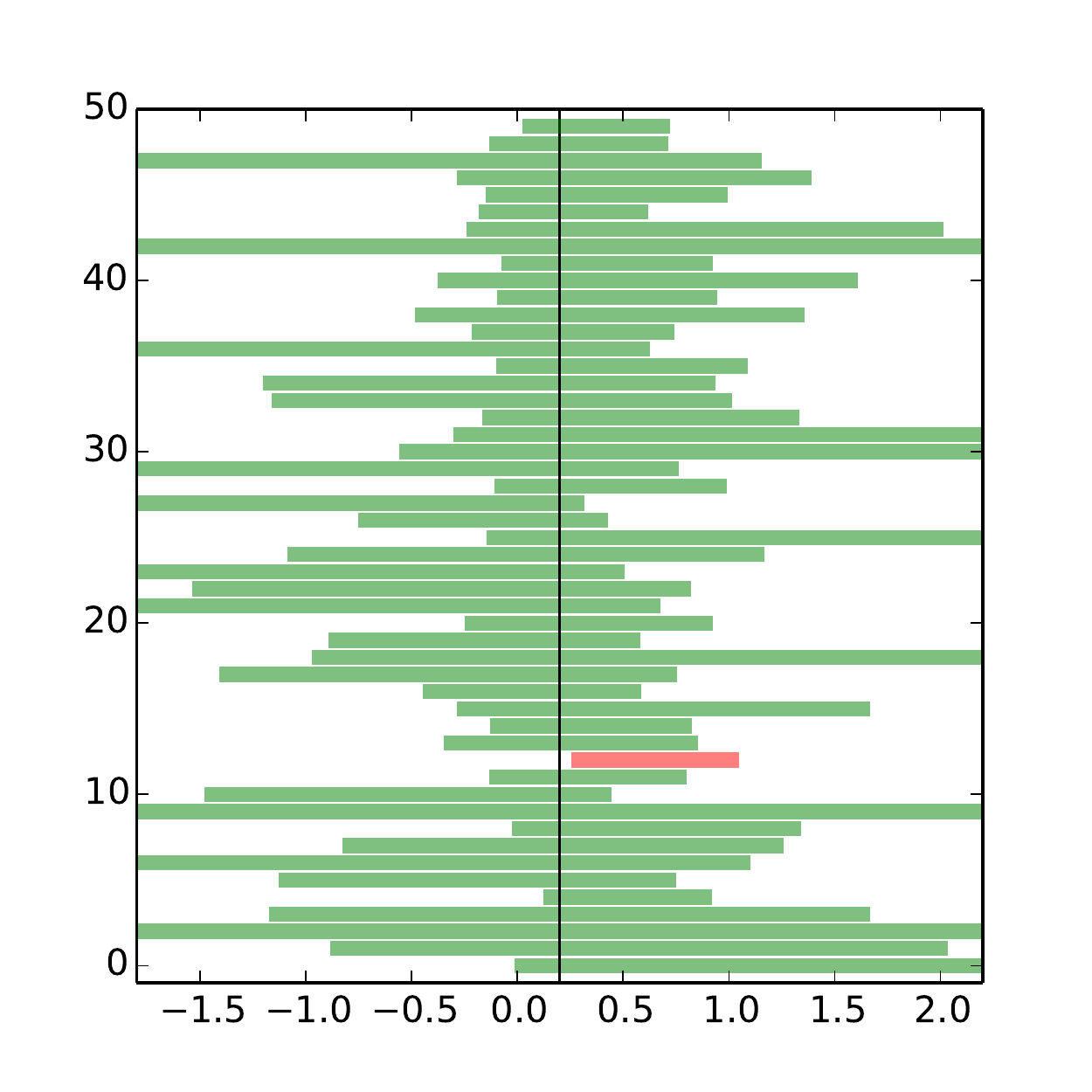}
\includegraphics[width = .48\textwidth]{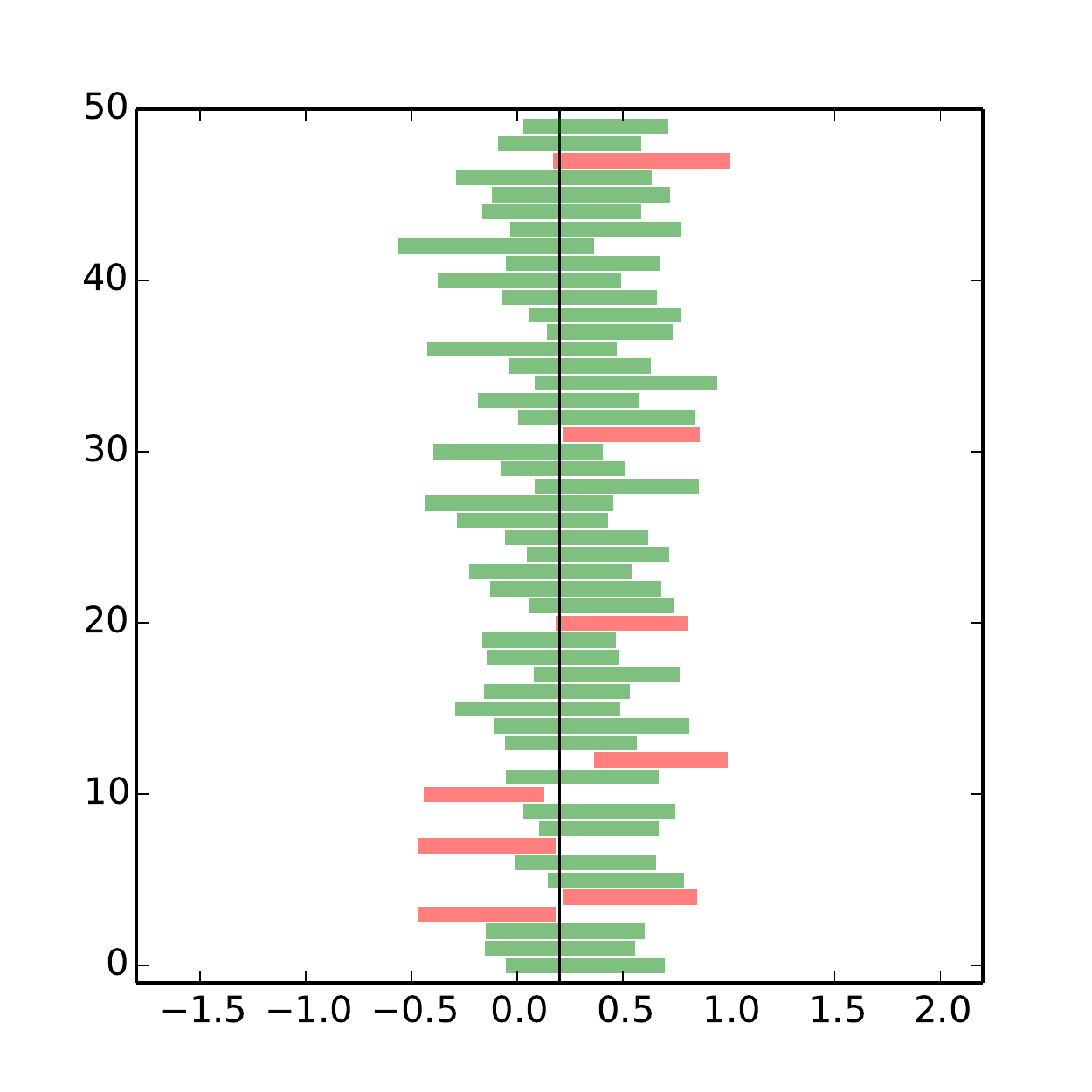}
\caption{95\% corrected (left) and normal (right) intervals for $\beta_{2,1}$ in a Type 3 Tobit model over 50 simulated data sets (n = 100, $p_1 = p_2$ = 10, $\sigma_1^2 = \sigma_2^2$ = 1.0, $\sigma_{12} = 0.5$, standard Gaussian design). A green bar means the confidence interval covers $\beta_{2,1}$ while a red bar means otherwise. Although the corrected intervals are wider, they cover $\beta_{2,1}$ (black line) at the nominal rate. Empirically, we observe the coverage of the normal intervals becomes worse as the signal strength decreases.}
\label{fig:tobit-3-intervals}
\end{figure}

\subsection{Type 2 Tobit model}

The Type 2 Tobit model is very similar to the Type 3 model, except the $y_1^*$ not observed:
\BEQ
\begin{aligned}
y_1^* &= X_1\beta_1 + \epsilon_1,\,\epsilon_1\sim\cN(0, I) \\
y_2^* &= X_2\beta_2 + \epsilon_2,\,\epsilon_1\sim\cN(0,\sigma_2^2I) \\
y_{2,i} &= \BCAS y_{2,i}^* & \text{if }y_{1,i}^* > 0 \\ 0 & \text{otherwise} \ECAS,\,i=1,\dots,n.
\end{aligned}
\label{eq:type-2-tobit-model}
\EEQ
The error terms $\epsilon_1$ and $\epsilon_2$ are again jointly normal with covariance $\sigma_{12}^2I$. Without loss of generality, we set $\sigma_1^2 = 1$. Type 2 is also called the \emph{sample selection model} and plays significant parts in the evaluation of treatment effects and program evaluation. To fit a Type 2, we rely on Heckman's two-step estimator (Algorithm \ref{alg:two-step-estimator-type-3-tobit}) except we skip step 2 (estimating $\beta_1$).
% The standard (Type 1) \eqref{eq:tobit-model} is a special cases of the Type 2 in which $X_1 = X_2$ and $\epsilon_1 = \epsilon_2$. However, unlike the Type 1 Tobit, the uncensored responses $\y_2$ may take negative values: $y_{1,i}^* > 0$ merely implies $y_{2,i} \ne 0$. 

The two-step estimate of $\beta_2$ is given by the first part of 
\[
\hat{\gamma}_2 = \hat{Z}_2^\dagger\y_2,\,\hat{Z}_2 = \BMAT \X_2 & \hat{\lambda} \EMAT,
\]
Given the censoring event $\y_1 > 0$, the pair $(\y_1,\y_2)$ has a constrained normal distribution, \ie{} the pdf of $\y = (\y_1,\y_2)$ is
\[
\BMAT \y_1 \\ \y_2\EMAT \sim \cN\left(\BMAT X_1\beta_1 \\ X_2\beta_2\EMAT,\BMAT I & \sigma_{12} I \\ \sigma_{12} I & \sigma_2^2 I\EMAT\right),\text{ subject to }\y_1 > 0.
\]
Let $\y = (\y_1,\y_2)$ and $\mubar,\Sigmabar$ be its (unconditional) expected value and covariance. To perform inference on $\beta_{2,j}$, we first express our target as
\[
\beta_{2,j} = e_j^T\X_2^\dagger\X_2\beta_2 = \underbrace{e_j^T\BMAT 0 & \X_2^\dagger\EMAT}_{\eta_{2,j}^T}\mubar.
\]

%However, when $\epsilon_1$ and $\epsilon_2$ are independent, 
%\[
%a = \frac{A\Sigmabar\eta_{2,j}}{\eta_{2,j}^T\Sigmabar\eta_{2,j}} = 0,
%\]
%with $A = \BMAT-I & 0\EMAT$. Thus $(\V_{\eta_{2,j}}^+(\y),\V_{\eta_{2,j}}^-(\y)) = (\infty,-\infty)$. By Corollary \ref{cor:pivotal-quantity}, we know \eqref{eq:pivotal-quantity-3}, which simplifies to
%\BEQ
%%F\left(\eta_{2,j}^T\y, \beta_{2,j},\eta_{2,j}^T\Sigmabar\eta_{2,j}, -\infty, \infty\right)
%\Phi\left(\bigl(\eta_{2,j}^T\y - \beta_{2,j}\bigr)\Big/\sqrt{\eta_{2,j}^T\Sigmabar\eta_{2,j}}\right),
%\label{eq:pivotal-quantity-2}
%\EEQ
%is uniformly distributed on the unit interval. To obtain confidence intervals for $\beta_{2,j}$, we simply invert \eqref{eq:pivotal-quantity-2} to obtain the usual normal intervals
%\[\textstyle
%\left[\eta_{2,j}^T\y + \Phi^{-1}\left(\frac{\alpha}{2}\right)\sqrt{\eta_{2,j}^T\Sigmabar\eta_{2,j}}, \eta_{2,j}^T\y + \Phi^{-1}\left(1-\frac{\alpha}{2}\right)\sqrt{\eta_{2,j}^T\Sigmabar\eta_{2,j}}\right].
%\]
%To test the significance of $\beta_{2,j}$, we evaluate \eqref{eq:pivotal-quantity-2} at $\beta_{2,j} = 0$ and reject when \eqref{eq:pivotal-quantity-2} is smaller than $\frac{\alpha}{2}$ or larger than $1-\frac{\alpha}{2}$. 

Since $\y_1$ is not observed, we cannot evaluate $\V_{\eta_{2,j}}^+(\y),\V_{\eta_{2,j}}^-(\y),\V_{\eta_{2,j}}^0(\y)$ to form the pivotal quantity \eqref{eq:pivotal-quantity-3}. To perform post-correction inference on $\beta_{j,2}$, we simulate the distribution of $\eta_{2,j}^T\y\mid\y_1 > 0$ with the parametric bootstrap \citep{efron1994introduction}. Since the pairs $(\y_{1,i},\y_{2,i})$ are independent, we simulate $\y\mid\y_1 > 0$ by simulating $n$ pairs
\[
\BMAT \y_1^* \\ \y_2^*\EMAT \sim \cN\left(\BMAT x_{1,i}^T\hat{\alpha}_1 \\ x_{2,i}^T\hat{\beta}_2\EMAT,\BMAT 1 & \hat{\sigma}_{12} \\ \hat{\sigma}_{12} & \hat{\sigma}_2^2\EMAT\right),\text{ subject to }\y_1^* > 0.
\]
Although simulating truncated normals is, in general, intensive, the fact that $\y_2^*$ is unconstrained, so the marginal distribution of $\y_1^*$ is a (univariate) truncated normal, allows us to simulate $(\y_1^*,\y_2^*)$ efficiently. To simulate a pair $(\y_1^*,\y_2^*)$, we first simulate $\y_1^*$ (with say the inverse CDF transform) and then simulate $\y_2^*$ conditioned on $\y_1^*$:
\BEQ
\begin{aligned}
\y_1^* &\sim TN(x_{1,i}^T\hat{\alpha}_1, 1, 0, \infty) \\
\y_2^* &\sim \cN(x_{2,i}^T\hat{\beta}_2 + \hat{\sigma}_{12}(\y_1^* - x_{1,i}^T\hat{\alpha}_1), \hat{\sigma}_2^2 - \hat{\sigma}_{12}^2).
\end{aligned}
\label{eq:simulate-bivariate-truncated-normal}
\EEQ
% It is easy to verify \eqref{eq:simulate-bivariate-truncated-normal} is correct (correct meaning the samples are distributed according to \eqref{eq:tobit-2-uncensored-responses}).

Given the bootstrap distribution of $\hat{\beta}_{j,2}$, it is straightforward to obtain confidence intervals and test the significance of $\beta_{j,2}$. We refer to \cite{efron1994introduction} for details. Figure \ref{fig:tobit-2-intervals} shows result from two simulations that compare the coverage of bootstrap versus normal intervals. We note the bootstrap intervals are comparable in size with the normal intervals.

\begin{figure}
\includegraphics[width = .48\textwidth]{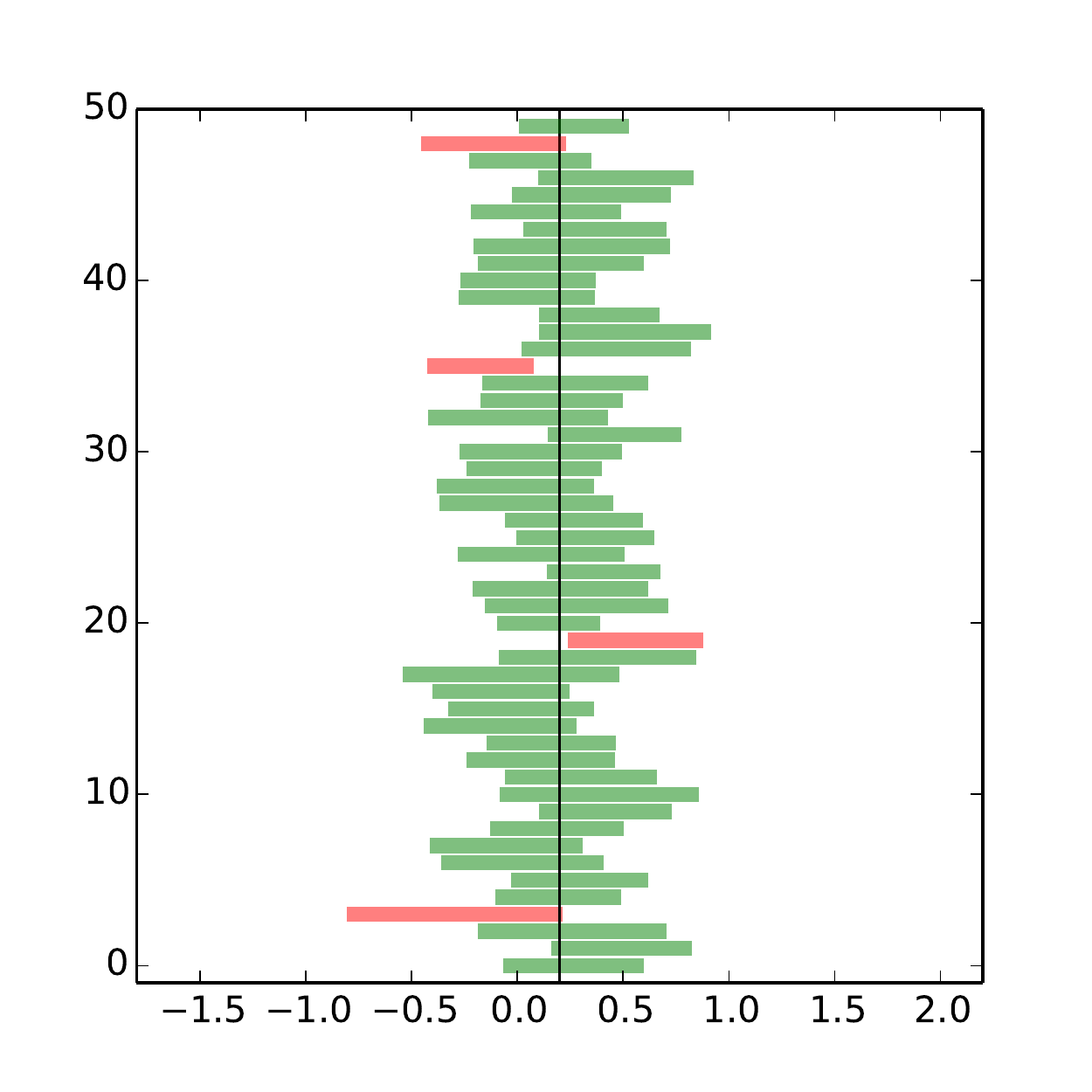}
\includegraphics[width = .48\textwidth]{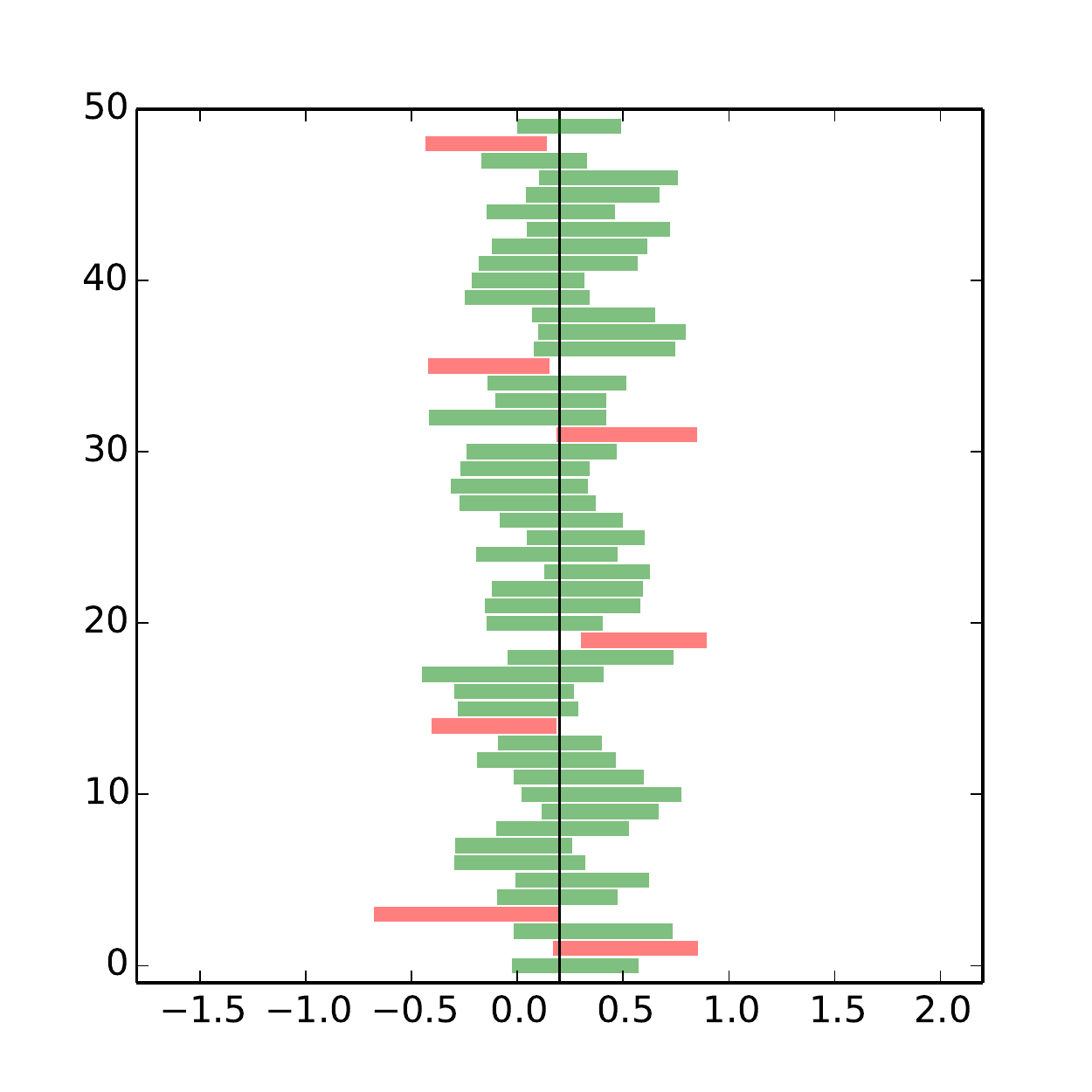}
\caption{95\% bias-corrected bootstrap (left) and normal (right) intervals for $\beta_{2,1}$ in a Type 2 Tobit model over 50 simulated data sets. The simulations are set up  like the simulations shown in Figure \ref{fig:tobit-2-intervals}. The bootstrap intervals are computed with 1000 bootstrap samples. A green bar means the confidence interval covers $\beta_{2,1}$ while a red bar means otherwise. }
\label{fig:tobit-2-intervals}
\end{figure}

\section{Summary and discussion}
\label{sec:summary}

We proposed a framework for conducting post-correction inference with two-step estimators. By conditioning on the censoring event, we obtain valid confidence intervals and significance tests for the fitted coefficients. We developed the framework with the standard (Type 1) Tobit model and showed how it generalizes to handle other censored regression models.

Although we must condition on the censoring event to perform inference, the results are valid \emph{unconditionally}. In Section \ref{sec:post-correction-inference}, we really derived a family of valid intervals/tests, one per possible outcome of the censoring mechanism. Given any outcome (censoring event), as long as we form the correct confidence interval, the interval will have the nominal coverage rate. In other words, the strategy of forming the correct interval given the censoring event inherits the validity of the conditional intervals. A similar strategy inherits the Type I error rate of the conditional tests. 

Our framework for performing post-correction inference may be readily combined with the framework of \cite{lee2013post} for post-(model) selection inference. In practice, one might wish to first select a model with the data and then perform inference on the selected model. For example, one might fit a Tobit model, observe which coefficients are significant at level $\alpha$, and report $1-\alpha$ confidence intervals for the significant coefficients. However, these intervals fail to account for the randomness in the selected model and may fail to cover the target at the nominal rate. By combining our framework with the framework of \cite{lee2013post}, it's possible to perform valid inference post-correction and post-(model) selection.

\section*{Acknowledgements} Will Fithian provided the geometric interpretation of Theorem \ref{thm:truncated-normal}. Y. Sun was partially supported by the NIH, grant U01GM102098. J.E. Taylor was supported by the NSF, grant DMS 1208857, and by the AFOSR, grant 113039.

\appendix
\section*{Appendix}
\subsection{Proof of Theorem \ref{thm:truncated-normal}}
\label{sec:truncated-normal-proof}

\begin{theorem}
Let $y\sim\cN(\mu,\Sigma)$. Define $a = \frac{A\Sigma\eta}{\eta^T\Sigma\eta}$ and
\begin{align*}
\V_\eta^+(y) &= \sup_{j:a_j < 0}\frac{1}{a_j}(b_j - (Ay)_j + a_j\eta^Ty) \\
\V_\eta^-(y) &= \inf_{j:a_j > 0}\frac{1}{a_j}(b_j - (Ay)_j + a_j\eta^Ty) \\
\V_\eta^0(y) &= \inf_{j:a_j = 0} b_j - (Ay)_j.
\end{align*}
Then, $\eta^Ty$ conditioned on $Ay \le b,\V_\eta^+(y), \V_\eta^-(y),\V_\eta^0(y) \ge 0$, has a truncated normal distribution, \ie{}
\[\textstyle
\eta^Ty\mid Ay \le b,\V_\eta^+(y), \V_\eta^-(y),\V_\eta^0(y) \ge 0 \sim TN(\eta^T\mu,\eta^T\Sigma\eta,\V_\eta^+(y),\V_\eta^-(y)).
\]
\end{theorem}

\begin{proof}
Our proof is similar in essence to the derivation in \cite{lee2013post}. First, we prove an auxiliary result that shows $Ay \le b$ implies $\{V^-(y) \le \eta^Ty\le V^+(y), V^0(y) \ge 0\}$ for some $\V_\eta^+(y),\V_\eta^-(y),\V^0(y)$ that are independent of $\eta^Ty$. 

\begin{lemma}
\label{lem:equivalent}
Let $y\sim\cN(\mu,\Sigma)$. Define $a = \frac{A\Sigma\eta}{\eta^T\Sigma\eta}$ and
\begin{align*}
\V_\eta^+(y) &= \sup_{j:a_j < 0}\frac{1}{a_j}(b_j - (Ay)_j + a_j\eta^Ty) \\
\V_\eta^-(y) &= \inf_{j:a_j > 0}\frac{1}{a_j}(b_j - (Ay)_j + a_j\eta^Ty) \\
\V_\eta^0(y) &= \inf_{j:a_j = 0} b_j - (Ay)_j.
\end{align*}
Then, $\{Ay \le b\}$ implies $\{\V_\eta^-(y) \le \eta^Ty\le \V_\eta^+(y), \V_\eta^0(y) \ge 0\}$. Further $\V_\eta^+(y),\,\V_\eta^-(y),\,\V_\eta^0(y)$ are independent of $\eta^Ty$.
\end{lemma}

\begin{proof}
The linear constraints $Ay \le b$ are equivalent to
\BEQ
Ay - \Expect[Ay\mid\eta^Ty] \le b - \Expect[Ay\mid\eta^Ty].
\label{eq:truncated-normal-1}
\EEQ
Since conditional expectation has the form
\[
\Expect[Ay\mid\eta^Ty] = A\mu + a(\eta^Ty - \eta^T\mu),\,a = \frac{A\Sigma\eta}{\eta^T\Sigma\eta},
\]
\eqref{eq:truncated-normal-1} simplifies to $Ay - b - a\eta^Ty \le -a\eta^Ty$. Rearranging, we obtain
\begin{align*}
\eta^Ty &\ge \frac{1}{a_j}(b_j - (Ay)_j + a_j\eta^Ty) & a_j < 0 \\
\eta^Ty &\le \frac{1}{a_j}(b_j - (Ay)_j + a_j\eta^Ty) & a_j > 0 \\
0 &\le b_j - (Ay)_j + a_j\eta^Ty & a_j = 0.
\end{align*}
We take the sup of the lower bounds and inf of the upper bounds to deduce
\[
\underbrace{\sup_{j:a_j < 0}\frac{1}{a_j}(b_j - (Ay)_j + a_j\eta^Ty)}_{\V_\eta^-(y)} \le \eta^Ty \le \underbrace{\inf_{j:a_j > 0}\frac{1}{a_j}(b_j - (Ay)_j + a_j\eta^Ty)}_{\V_\eta^+(y)}.
\]
Since $y$ is normal, $b_j - (Ay)_j + a_j\eta^Ty,\,j=1,\dots,m$ are independent of $\eta^Ty$. Hence $\V_\eta^+(y),\,\V_\eta^-(y),\,\V_\eta^0(y) \ge 0$ are also independent of $\eta^Ty$.
\end{proof}

To complete the proof of Theorem \ref{thm:truncated-normal}, we must show $\eta^Ty$ given $Ay \le b,\V_\eta^+(y), \V_\eta^-(y),\V_\eta^0(y) \ge 0$ is truncated normal. By Lemma \ref{lem:equivalent}, the conditional law of $\eta^Ty$ is
\begin{align*}
&F(\eta^Ty\mid Ay \le b,\V_\eta^+(y),\V_\eta^-(y),\V_\eta^0(y)\ge 0) \\
&\pc= F(\eta^Ty\mid\V_\eta^-(y) \le \eta^Ty\le\V_\eta^+(y),\V_\eta^+(y), \V_\eta^-(y),\V_\eta^0(y) \ge 0) \\
&\pc= \frac{F(\eta^Ty,\,v^- \le \eta^Ty\le v^+\mid\V_\eta^+(y) = v^+, \V_\eta^-(y) = v^-,\V_\eta^0(y) \ge 0)}{F(v^- \le \eta^Ty\le v^+\mid\V_\eta^+(y) = v^+, \V_\eta^-(y) = v^-,\V_\eta^0(y) \ge 0)},
\end{align*}
Since $\V_\eta^+(y),\,\V_\eta^-(y),\,\V_\eta^0(y)$ are independent of $\eta^Ty$, 
\begin{align*}
&F(\eta^Ty\mid Ay \le b,\V_\eta^+(y) = v^+, \V_\eta^-(y) = v^-,\V_\eta^0(y)\ge 0) \\
&\pc = \frac{F(\eta^Ty,\,v^- \le \eta^Ty\le v^+)}{F(v^- \le \eta\le v^+)},
\end{align*}
which is the CDF of $TN(\eta^T\mu,\eta^T\Sigma\eta,\V_\eta^+(y),\V_\eta^-(y))$
\end{proof}

\subsection{Monotonicity of $F$}
\label{sec:F-monotone}

\begin{lemma}
\label{lem:F-monotone}
Let $F(x,\mu,\sigma^2,a,b)$ be the CDF of a truncated normal random variable. Then $F$ is monotone decreasing in $\mu$.
\end{lemma}

\begin{proof}
The truncated normal distribution is a natural exponential family in the mean $\mu$. Thus, its likelihood ratio is monotone in $\mu$, \ie{} for $\mu_1 < \mu_2$ and $y_2 < y_2$,
\[
\frac{f_{\mu_2}(y_1)}{f_{\mu_2}(y_2)} < \frac{f_{\mu_2}(y_2)}{f_{\mu_1}(y_2)}.
\]
This implies $f_{\mu_2}(y_1)f_{\mu_1}(y_2) < f_{\mu_2}(y_2)f_{\mu_1}(y_2)$. We integrate with respect to $y_1$ over $(-\infty,y]$ and with respect to $y_2$ over $[y,\infty)$ to obtain
\[
(1 - F_{\mu_1}(y))F_{\mu_2}(y) < (1 - F_{\mu_2}(y))F_{\mu_1}(y).
\]
We subtract the cross-terms to conclude $F_{\mu_2}(y) < F_{\mu_1}(y)$.
\end{proof}

\bibliographystyle{imsart-nameyear}
\bibliography{censored}

\end{document}

%% file: macros.tex
\newcommand{\BALD}{\begin{aligned}}
\newcommand{\EALD}{\end{aligned}}
\newcommand{\BALDS}{\begin{aligned*}}
\newcommand{\EALDS}{\end{aligned*}}
\newcommand{\BCAS}{\begin{cases}}
\newcommand{\ECAS}{\end{cases}}
\newcommand{\BEAS}{\begin{eqnarray*}}
\newcommand{\EEAS}{\end{eqnarray*}}
\newcommand{\BEQ}{\begin{equation}}
\newcommand{\EEQ}{\end{equation}}
\newcommand{\BIT}{\begin{itemize}}
\newcommand{\EIT}{\end{itemize}}
\newcommand{\BMAT}{\begin{bmatrix}}
\newcommand{\EMAT}{\end{bmatrix}}
\newcommand{\BNUM}{\begin{enumerate}}
\newcommand{\ENUM}{\end{enumerate}}
\newcommand{\eg}{e.g.}

\newcommand{\ie}{i.e.}

\newcommand{\BA}{\begin{array}}
\newcommand{\EA}{\end{array}}

\newcommand{\ones}{\mathbf 1}

\newcommand{\reals}{\mathbf{R}}

\DeclareMathOperator*{\argmax}{\arg\max}
\DeclareMathOperator*{\argmin}{\arg\min}

\DeclareMathOperator{\Expect}{\mathbf{E}}

\DeclareMathOperator{\Prob}{\mathbf{Pr}}

\DeclareMathOperator{\var}{var}

\newcommand{\pc}{\hspace{1pc}}

\newcommand{\norm}[1]{\left\| #1 \right\|}